\documentclass[sort&compress]{elsarticle}

\usepackage[utf8]{inputenc}
\usepackage[T1]{fontenc}
\usepackage[USenglish]{babel}
\usepackage{bbm}
\usepackage{graphicx}
\usepackage{upref}
\usepackage{amsmath}
\usepackage{amssymb}
\usepackage{amsthm}
\usepackage{microtype}
\usepackage{subfigure}
\usepackage{lmodern}

\usepackage[
 colorlinks = false, 
 backref,
 pdfauthor={Karl Bringmann, Kurt Mehlhorn, Adrian Neumann},
 pdftitle={Remarks on Category-Based Routing in Social Networks},
 unicode]
 {hyperref}

\def\clap#1{\hbox to 0pt{\hss#1\hss}}

\def\mathrlap{\mathpalette\mathrlapinternal}

\def\mathrlapinternal#1#2{%
\rlap{$\mathsurround=0pt#1{#2}$}}

\newcommand{\calS}{\mathcal{S}}
\newcommand{\Oh}{\mathcal{O}}
\newcommand{\st}{\mid}
\newcommand{\floor}[1]{\left\lfloor #1 \right\rfloor}
\newcommand{\ceil}[1]{\left\lceil #1 \right\rceil}
\newcommand{\abs}[1]{\left| #1 \right|}
\newcommand{\set}[2]{\{ #1 \st #2 \}}

\usepackage{color}
\newcommand{\comment}[1]{}

{\bfseries}{\itshape}
\newtheorem{lemma}{Lemma}{\bfseries}{\itshape}
\newtheorem{definition}{Definition}{\bfseries}{\itshape}
\newtheorem{corollary}{Corollary}{\bfseries}{\normalfont}
\newtheorem{theorem}{Theorem}

\DeclareMathOperator*{\memdim}{memd}
\DeclareMathOperator*{\diam}{diam}
\DeclareMathOperator*{\lnln}{lnln}
\DeclareMathOperator*{\loglog}{loglog}
\DeclareMathOperator*{\dist}{dist}
\DeclareMathOperator*{\cdist}{d}
\DeclareMathOperator*{\cat}{cat}
\title{Remarks on Category-Based Routing
in Social Networks}
\author{Karl Bringmann}
\author{Kurt Mehlhorn}
\author{Adrian Neumann}
\address{Max-Planck-Institut f\"{u}r Informatik, Campus E14, 66123 Saarbr\"{u}cken, Germany}

\begin{document}
\begin{frontmatter}

\begin{abstract}
	 It is well known that individuals can route messages on short
paths through social networks, given only simple information about the target
and using only local knowledge about the topology. Sociologists conjecture that people find routes greedily by passing the message to an acquaintance that has
more in common with the target than themselves,
e.\,g.\ if a dentist in Saarbr\"{u}cken wants to send a message to a specific lawyer in Munich, he may forward it to someone
who is a lawyer and/or lives in Munich. Modelling
this setting, Eppstein et~al.\ introduced the notion of \emph{category-based
routing}. The goal is to assign a set of categories to each node of a graph
such that greedy routing is possible. By proving bounds on the number of categories a node has to be in we can argue about the plausibility of the underlying sociological model. In this paper we substantially improve the upper bounds introduced by Eppstein et~al.\ and prove new lower bounds.
\end{abstract}

\end{frontmatter}

\section{Introduction}

In the 1960's, Milgram~\cite{Milgram_1967, Milgram_1969, Korte_Milgram_1970} observed the \emph{small world phenomenon}, i.\,e., that short paths  seem to connect us all in the social graph in which the vertices represent persons and two persons are connected by an edge if and only if they know each other. He performed experiments in which he asked randomly selected participants to relay letters across the USA by passing them to one of their direct acquaintances.

The participants only had simple categorical information about the target, such as  name, location, and occupation, and knowledge about their own acquaintances, that is, they knew the local topology of the network. The  messages arrived typically after only six hops.

Perhaps even more surprising than the mere existence of short paths in social networks is the apparent ease with which humans can discover them, despite having only little information. Experiments by sociologists~\cite{social_1988, social_1977, social_1974, social_2002} indicate that we use a simple greedy heuristic to route. The message is passed on to the acquaintance that is most similar to the target, where two persons are similar if they share common characteristics, e.\,g.\ they have the same job or live in the same city.

Graph-theoretic models enable us to check the plausibility of these theories. If it turns out that greedy routing needs strong assumptions about the underlying social structure or requires a complex distance function this would be evidence that the actual mechanism humans use is different. 

In~\cite{orig_paper}, Eppstein et~al.\ model this setting as a connected graph $G=(V,E)$,
where people act as vertices and whose edges represent pairwise acquaintance,
together with a system of categories $\calS \subseteq \mathcal{P}(V)$. Each category $C \in \calS$ is the vertex set of a connected subgraph of $G$. For two
vertices $s$ and $t$, let the distance from $s$ to $t$ be the number of
categories containing $t$ but not $s$, i.\,e.,
\[ 
	d(s,t) = \abs{\set{C \in \calS}{t \in C \text{ and } s \not\in C}}. 
\]
A system of categories $\calS$ supports \emph{greedy routing} in $G$ (is
\emph{good} for $G$) if for any
two vertices $s$ and $t$, there is a neighbor $u$ of $s$ with $d(u,t) <
d(s,t)$. 
For a system of categories $\calS$, Eppstein et~al.\ define its
\emph{membership dimension} as the maximal number of categories to which any
vertex belongs:
\[
 \memdim(\calS) := \max_{v\in V} \abs{\{C \in \calS \st v \in C\}}.
\]
The goal is then to (constructively) show the existence of a system of categories that supports greedy routing and has a small
membership dimension. Membership dimension captures the cognitive load of the
participants, i.\,e., the number of categories an actor must keep track of in
order to decide on the next node of the route. If the required membership dimension is too high, humans likely use a different method to find routes.

Eppstein et~al.~\cite{orig_paper} show
the existence of a good system of categories $\calS$ with
\[
	\memdim(\calS) \in \Oh\left((\diam(G) + \log \abs{V})^2\right)
\]
 and
note a lower bound of $\diam(G)$; here $\diam(G)$ denotes
the diameter of $G$. 

We substantially improve
on the upper bound stated above and establish new 
lower bounds. We review related work in
Sect.~\ref{sec:related_work}, introduce notation in
Sect.~\ref{sec:notation}, and prove exact bounds for lines, grids, and tori
in Sect.~\ref{sec:simple_bounds}. In
Sect.~\ref{sec:improved_upper_bound}, we construct for every graph $G$ a good system of categories
$\calS$ with
\[
	\memdim(\calS) = \Oh\left(\diam(G) \cdot \log\left(\frac{2\abs{V}}{\diam(G)}\right)\right).
\]
This bound improves upon the bound of Eppstein et~al.\ except for $\diam(G) = \Theta(\log \abs{V})$. In Sect.~\ref{sec:stars_and_trees} we show an almost matching lower bound. 
We exhibit for all nonnegative integers $n$ and $d$ a
graph with $1 + nd$ vertices and diameter $2d$ for which every good system of
categories has membership dimension
\[
	\Omega\left(\frac{d \ln (\abs{V}\!/d)}{\ln \left(d \ln (\abs{V}\!/d) \right)}\right).
\]
In Sect.~\ref{sec:general_lower_bound} we show that every 
good system $\calS$ for a graph $G$ with average degree $\delta$
has membership dimension
\[
    \Omega\left(\diam(G) +\frac{\log \abs{V}}{\log \delta}\right),
\]
in particular, bounded degree graphs require logarithmic membership dimension. The bound is best possible. For each triple 
$(n,\delta,\diam)$ of positive reals with $1 \le \delta, \diam \le n$ we exhibit a graph $G$ with
$\Theta(n)$ vertices, average degree $\Theta(\delta)$, and diameter $\Theta(\diam)$ for
which a good system of membership dimension
\[
         \Oh\left(\diam(G) +\frac{\log \abs{V}}{\log \delta}\right).
\]
exists.

\section{Related Work}\label{sec:related_work}

Greedy Routing is a well-studied technique with many applications in computer science. A variety of methods is known; for example geographical information as an aid for routing has been explored in~\cite{Finn, Kranakis99compassrouting}. This method does not succeed on all networks, so a number of enhancements have been developed ~\cite{springerlink:10.1023/A:1012319418150, Karp:2000:GGP:345910.345953, Kuhn:2003:GAR:872035.872044}. Several groups also examine \emph{succinct} greedy-routing strategies that limit the additional information at every vertex to be logarithmic in the size of the network~\cite{springerlink:10.1007/978-3-642-00219-9-3, springerlink:10.1007/978-3-642-10631-6-79, Maymounkov06greedyembeddings, 4151823}.

Category-based routing can be seen as a special case of succint greedy routing. One method similar to our framework assigns each node a point in a metric space of low dimension. Each node passes a message to the neighbour with the lowest distance to the target. In this setting very good bounds can be achieved. For example Flury et~al.~\cite{flury} give a construction that not only requires only polylogarithmic dimensionality, but also provides logarithmic bounds on the stretch, i.e.\ the factor by which the greedy routes are longer than the shortest paths.

However, our model is more restrictive as we don't consider routing between computers but want to investigate a natural mechanism for message passing between humans. We adopt the model of~\cite{orig_paper}, who were the first to study category-based routing from complexity-theoretic point of view. We reviewed their results in the introduction.

A different approach to routing in social networks was studied by Kleinberg~\cite{4215803}. He focuses on location instead of categorical information to explain how we find short routes efficiently. Based on his insights he constructs a random graph model that has similar properties to real world networks and shows for which parameters routing is possible. In contrast to this, the approach of Eppstein et~al.\ seeks to construct a system of categories that enables greedy routing for a given network.

The problem has been investigated in the social sciences. There have a number of experimental studies, for example Killworth and Bernard~\cite{social_1977} show that humans use categorical information, foremost location and occupation, for finding routes. Dodds et~al.~\cite{social_2003} show that professional relationships are important for deciding on the next hop. There have also been efforts to model the category structure of social networks, resulting in even more restrictive models than that of Eppstein et~al. For example Watts et~al.~\cite{social_2002} define a model for social networks in which nodes are grouped in a hierarchy of categories in which each category contains only a small number of nodes. As humans are more likely to befriend persons similar to themselves, the probability that a connection exists between two nodes depends on their similarity. They define a complicated distance function that models the way humans judge similarity and show experimentally that greedy routing succeeds over a wide range of the tunable parameters in their model.

\section{Preliminaries}\label{sec:notation}

Unless stated otherwise, we consider undirected connected graphs $G=(V,E)$ with $n=\abs{V}$
nodes. For two nodes $u\in V$, $v\in V$ let $\dist(u,v)$
be the number of edges on a shortest path between $u$ and $v$. Then define the
diameter of a graph $G$ as 
\[
 \diam(G) := \max_{u\in V,v\in V} \dist(u,v),
\]
that is, the diameter is the length of a longest shortest path in $G$.

Let $\calS\subseteq \mathcal{P}(V)$ be a system of subsets of the vertices
of
$G$ that induce connected subgraphs. For a node $u\in U$ define $\cat(u)$ to be the set of
categories to which $u$ belongs,
\[
 \cat(u) :=  \{ C\in \calS \st u\in C \} .
\]
The membership dimension $\memdim(\calS)$ is the maximum number of categories to which any node belongs, that is,
\[
 \memdim(\calS) := \max_{u \in V} \abs{\cat(u)}.
\]

If the message addressed to node $t\in V$ is currently in node $u\in V$ and $u\neq t$, the algorithm forwards it to a neighbor $v\in N(u)$ of $u$ that is closer to $t$ according to the distance function
\[
 \cdist(v,t) := \abs{\cat(t) \backslash \cat(v)},
\]
The algorithm succeeds if for all $u, t\in V$, $u \ne t$, there is a neighbor $v\in N(u)$ such that
$\cdist(v,t) <  \cdist(u,t)$. We say a system of categories supports greedy routing in $G$ or is \emph{good} for $G$ if a greedy routing algorithm succeeds.

\section{Simple Bounds}\label{sec:simple_bounds}

Already in Eppstein et~al.~\cite{orig_paper}, it is observed that the diameter bounds the membership dimension from below.
\begin{lemma}\label{lem:path_lower_bound}
 For any graph $G$ and good system of categories $\calS$ we have
\[ \memdim(\calS) \geq \diam(G).\]
\end{lemma}
\begin{proof}
 Let $s\in V(G)$, $t \in V(G)$ be a pair of vertices such that $\dist(s,t)=\diam(G)$. Consider the path $P$ that a message from $s$ to $t$ takes according to the greedy routing system. Note that $\abs{P} \geq \diam(G)$. For every edge $(u,v)\in P$ we have $\cdist(v,t)\leq \cdist(u,t)-1$,
and hence 
\[
	\cdist(s,t) \geq \abs{P} \geq \diam(G).
\]
 By the definition of $\cdist(\cdot)$, this is only possible if $|\cat(t)| \geq \diam(G)$.
\end{proof}

For paths this bound is tight, cf.\ the construction in
Fig.~\ref{fig:path_categories}. We can extend this observation to all graphs
that can be obtained from paths and cycles by taking cross products. 
\begin{figure}[t]
 \centering
 \includegraphics[width=0.5\linewidth]{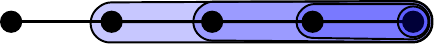}
 \caption{Minimal categories for routing from the leftmost vertex to the
rightmost. Adding the symmetric categories allows to route between any pair of
vertices. Every vertex is contained in exactly four categories.}
 \label{fig:path_categories}
\end{figure}

\begin{definition} Let $G=(V,E)$, and $H=(V',E')$, be two graphs. Then the \emph{cross product} $G\times H$ is the graph $(\tilde V, \tilde E)$ with
 \begin{align*}
	\tilde V &= V \times V',\\
	\tilde E &= \{\{(u,x),(v,y)\} \st (\{u,v\}\in E\phantom{'} \wedge x=y) \vee (\{x,y\} \in E' \wedge u=v)\}).
\end{align*}
\end{definition}

\begin{lemma}\label{lem:cross_product}
 Let $M(G)$ be the minimal membership dimension needed to route in $G$. Then 
\[
	M(G\times H) \le M(G)+M(H).
\]
\end{lemma}

\begin{proof}
 Let $C_G$ be a minimal category system for $G$ and let $C_H$ be a minimal
category system for $H$. We construct a system of membership dimension
$M(G) + M(H)$ for $G \times H$. For every $c
\in C_G$, we add $c \times V'$ to the system and for every $c' \in C_H$, we add
$V \times c'$ to the system.

We route from $(u,x)$ to $(v,y)$ by initially using the first kind of
categories 
to route to $(v,x)$ and then using the second kind of categories to route to $(v,y)$.
\end{proof}

Lemmas \ref{lem:path_lower_bound} and \ref{lem:cross_product} immediately give the following tight bounds.

\begin{corollary}
 For grid graphs $G$ we get $M(G) = \diam(G)$. \qed
\end{corollary}
\begin{corollary}
 For hypercubes $G$ we get $M(G) = \diam(G)=\ceil{\log n}$. \qed
\end{corollary}
\begin{corollary}
 For tori $G=C_k\times C_l$ we have $M(G) = \ceil{k/2} + \ceil{l/2}$. \qed
\end{corollary}

\section{Improved Upper Bound}\label{sec:improved_upper_bound}

We construct for every graph $G$ a good system of categories $\calS$ with
membership dimension 
\[ \Oh\left( \diam(G) \cdot \log \frac{2 \abs{V}}{\diam(G)}\right).\]
It suffices to prove the bound for trees. For general graphs, we construct a
spanning tree of diameter $\diam(G)$ and route in the spanning tree. 

\begin{lemma}\label{lem:trees}
For any tree on $n$ nodes with diameter $d$ there is a system of categories $\calS$ of membership dimension 
\[
 \memdim(\calS) = \Oh(d \log(2 n/d)).
\]
\end{lemma}

We conjecture this bound to be tight. An
example might be the star with degree 
$n/d$ sending out paths of length $d$ as shown in Fig.~\ref{fig:long_star}.

\begin{proof}
For a tree $T$ let $r\in V(T)$ and consider a triple $(r,L,R)$ with
$L,R$ a partitioning of the neighbors of $r$ in $T$. 
After deleting $r$, the tree falls into components; take the ones containing
nodes from $L$. To these components add $r$ again (making it a neighbor to all
nodes in $L$) to get a tree $T_L = T_L(r)$. Build $T_R$ symmetrically. This
cuts $T$ into two trees; each vertex is in exactly one of $T_L$ and $T_R$
except for $r$, which belongs to both trees.  

Now, we call $(r,L,R)$ a \emph{balanced routing cut} if both $\abs{V(T_L)}$ and
$\abs{V(T_R)}$ are at least $c_1 \cdot \abs{V(T)}$ and at most $c_2 \cdot
\abs{V(T)}$.  
A basic graph theoretic argument shows that any tree has a balanced routing cut
for constants $c_1 = 1/3$ and $c_2 = 2/3$, see for example~\cite{chung89}.

Given a tree $T$ we now construct categories as follows. Take a balanced
routing cut $(r,L,R)$ of $T$. Similar to Eppstein et al.~\cite{orig_paper}, we
construct categories that allow routing from any vertex in $T_R$ all the way to
$r$ when having a target in $T_L$ (and symmetrically), see
Fig.~\ref{fig:tree_cat_1}. Then we recurse on $T_R$ and $T_L$. After that, we
unify some categories to decrease the membership dimension.\footnote{The recursion
means that we again find a balanced routing cut in $T_R$, so the vertex we
split at in $T_R$ does not have to be $r$ or in $R$. This is a major difference
to Eppstein et al.} 

\begin{figure}[t]
 \centering
 \includegraphics{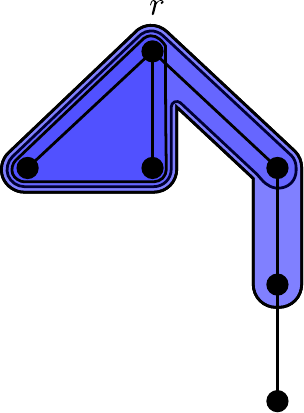}
 \caption{Categories that allow routing from $T_R$ to $r$ for a target in $T_L$.}
 \label{fig:tree_cat_1}
\end{figure}

The base case of this procedure is a graph of constant size, where we add any
valid system of categories of constant size. Observe that this way we construct
a valid system of categories for $T$. 

We now describe this construction in more detail.
Let $d_v = \dist(v,r)$ be the distance from $v \in V(T)$ to $r$ in
$T$.\footnote{or in $T_R$ or $T_L$, there is no difference if $v$ is in that
tree.} For routing from $T_R$ to $T_L$ we add the categories $V(T_L) \cup \{v \in V(T_R) \mid d_v \le k\}$,
for $0 \le k \le \diam(T_R)$.\footnote{This is the same construction as in
Eppstein et al.~\cite{orig_paper}} We add symmetric categories for routing from
$T_L$ to $T_R$. Note that these categories allow us to route from any vertex in $T_R$ to $r$
when having a target in $T_L$. 

Then we recurse on $T_R$ and $T_L$.

\begin{figure}[t]
\centering

 \subfigure[{The categories before recursing.}]{
	\makebox[0.45\linewidth]{\label{fig:before_recurse}\includegraphics{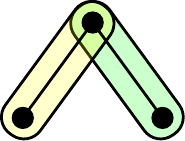}}}
 \hfill
 \subfigure[The categories for two nodes.]{
	\makebox[0.45\linewidth]{\label{fig:during_recurse}\includegraphics{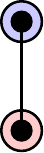}}}
 \subfigure[Categories after recursing.]{
	\makebox[0.45\linewidth]{\label{fig:after_recurse}\includegraphics{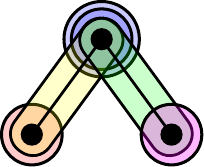}}}
 \hfill
 \subfigure[The merged categories.]{
	\makebox[0.45\linewidth]{\label{fig:merged}\includegraphics{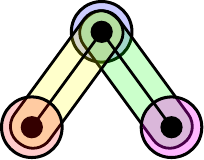}}}

\caption{Merging of categories after recursing.}
\label{fig:merging}
\end{figure}

After that, we change the system of categories slightly to decrease
its membership dimension. We have split at $(r,L,R)$ and in the recursion on $T_R$ we
constructed, say, categories $R_1,\ldots,R_a$ containing $r$ and in the
recursive call to $T_L$ we constructed categories $L_1,\ldots,L_b$ containing
$r$. Assume wlog.\ $a \le b$. Then we can replace the $R_i$ and $L_j$ by the
categories  
\[
	\{R_i \cup L_i \st 1\leq i\leq a\} \cup \{L_{a+1}, \ldots, L_b\}.
\]
These categories are still connected and greedy routing is still possible, as
the $R_i$ ($L_j$) were only needed to route inside $T_R$ ($T_L$). See
Fig.~\ref{fig:merging}. 

This ends the construction of categories. Observe that we construct a valid system of categories.

It remains to bound the membership dimension of the constructed system of
categories. Consider any node $v \in V(T)$. Joining categories as described in
the third part of our construction implies that the number of categories
containing $v$ equals the maximal number of categories containing $v$
constructed on any path in the recursion tree. As we select balanced routing
cuts, in each recursive call of the construction the number of nodes of the
subproblems is decreased by a constant factor and hence each such path has
length $\Oh(\log(n))$.  

Consider one such path. On level $i$, $0\le i \le \Oh(\log(n))$, we considered a subproblem $T_i$ on at most $n \cdot c_2^i$ nodes. The diameter of $T_i$ is bounded from above by $d = \diam(T)$, as $T_i$ is a subtree of $T$. Moreover, the diameter of $T_i$ is bounded by its number of nodes. Hence, 
\[
	\diam(T_i) \le \min\{ d, n \cdot c_2^i \}.
\]
Our procedure cuts $T_i$ and adds some categories to allow routing over that cut. Observe that the number of such categories is bounded by $\Oh(\diam(T_i))$. Thus, the number of categories containing $v$ can be bounded from above by (omitting the $\Oh()$)
\[
\sum_{i=0}^{\Oh(\log(n))} \min\{ d, n \cdot c_2^i \}.
\]
We can bound this sum as
\[
\sum_{i=0}^{-\log_{c_2}(n/d)} d \quad + \quad \quad \sum_{\mathrlap{i=-\log_{c_2}(n/d)+1}}^{\Oh(\log(n))} n \cdot c_2^i,
\]
which simplifies to 
\[
	d \log(n/d) + d = \Oh(d \log(2n/d)).
\]
\end{proof}

\section{Stars}\label{sec:stars_and_trees}

In this section we prove upper and lower bounds for stars. 
The lower bound for stars nearly matches the upper bound of the preceding
section. The star of diameter $2d$ and $\ell$ leaves has $1 + \ell d$ nodes. The center
node has degree $\ell$ and each leaf is joined to the center node by a path of
length $d$. Fig.~\ref{fig:star} shows a star with four leaves and diameter
2. Note that as $|V| = 1 + \ell d$ for stars we can replace $\ell$ by $n/d$ in our asymptotic bounds.
We start with simple upper and lower bounds.

\begin{lemma}\label{lem:simple_star_lower_bound} In a star with $\ell$ leaves, the center node is contained in at least
$\log \ell$ categories. \end{lemma}
\begin{proof} It suffices to show the claim for stars with diameter 2, 
as every star with higher diameter includes one with diameter 2.
Let $C_1,\ldots,C_m$ be the categories containing the center vertex $c$ and consider, for each leaf $u$,
the bitstring $p^{(u)}$ of length $m$ defined by 
\[
   \left( p^{(u)} \right)_i = (u \in C_i).
\]
If two bitstrings $p^{(u)}, p^{(v)}$ are equal for leafs $u \ne v$, greedy routing from $u$
to $v$ is impossible, as the distance to $v$ does not decrease along the edge $u c$.
Hence, these $\ell$ bistrings are pairwise different.
Since there are only $2^m$ different bitstrings of size $m$, this implies
$\cat(c) = m \ge \log \ell$.

An alternative proof can be found in the Appendix, Proof~\ref{pr:simple_star_lower_bound}.
\end{proof}

The same argument establishes:

\begin{lemma} In a tree, each node $v$ is contained in at least $\log \deg v$
categories. \qed
\end{lemma}

\begin{lemma}\label{lem:simple_star_upper_bound} For a star with $\ell$ leaves and diameter $2$, there is a system of categories $\calS$ with membership dimension
\[
 \memdim(\calS) \leq 1 + 2 \ceil{\log \ell}.
\]
\end{lemma}
\begin{proof} Let the leaves be numbered from $1$ to $\ell$. Every leaf forms a category of its own. For every $i$, $0 \le i < \ceil{\log \ell}$ we have two 
categories $Z_i$ and $O_i$: $Z_i$ contains the center and all leaves that  have
a zero in the $i$-th bit of their binary representation; $O_i$ contains the
center and all leaves that have a one in the $i$-th bit of their binary
representation. See Fig.~\ref{fig:star}. Clearly every node is contained in at most $1+2\ceil{\log \ell}$ categories. 
\begin{figure}[t]
	\centering
 	\includegraphics[width=0.5\linewidth]{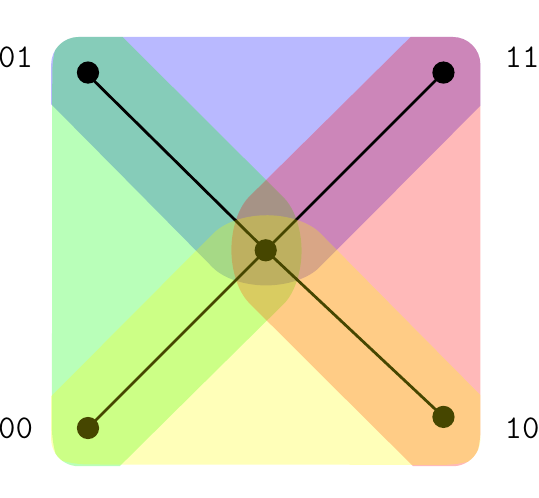}
	\caption{Categories $Z_i$ and $O_i$ for a star with $\ell=4$.}\label{fig:star}
\end{figure}

Consider any two leaves $u$ and $v$ and assume we want to route from $u$ to
$v$. Let $k$ be the number of positions in which the binary representation of $u$ and $v$ differ. Then $\cdist(u,v) = 1 + k \ge 2$ and $\cdist(c,v) = 1$. Thus we can successfully route from $u$ to
$v$. 
\end{proof}

We next improve on both bounds.

\begin{lemma}\label{lem:tight_star_bound} Let $k$ be minimal such that
 \[
  {k \choose \floor{k/2}} \geq \ell.
 \]

Then $k = \log \ell + (\frac{1}{2} + o(1)) \log\log \ell$.

There is a system of categories $\calS$ that supports routing in stars with $\ell$
leaves and diameter $2$ having
\[
 \memdim(\calS) = k.
\]
This bound is tight, i.e., every good system of categories $\calS$ for this graph
has
\[
 \memdim(\calS) \ge k.
\]
\end{lemma}
\begin{proof}
Let $c$ be the center of the star. For any two distinct leaves $i$ and $j$ there must be categories $C$ and $C'$ such that $C,C' \in \cat(j)
\setminus \cat(i)$, $C \in \cat(c)$, and $C' \not\in \cat(c)$, as otherwise one cannot
route from $i$ to $j$. Let $C_1$, \ldots, $C_k$ be the categories containing
the center. For every $i$ define a binary vector $v_i$ of length $k$ by
\[      
	(v_i)_h = (i \in C_h) \quad\text{for $1 \le h \le k$.} 
\]
The vectors $v_i$, $1 \le i \le \ell$ form an anti-chain in the set of all binary
vectors of length $k$ as for every $i$ and $j$ with $i \neq j$, there must be
a $h$ with $(v_i)_h = 0$ and $(v_j)_h = 1$. As the maximal size of an anti-chain in
the set of all binary vectors of length $k$ is 
\[
 { k \choose \floor{k/2}},
\]
the lower bound follows.

We turn to the the upper bound. There are $\ell$ distinct bitstrings of length $k$ each containing exactly $\floor{k/2}$
ones. Arbitrarily assign the strings to the leaves. We have $k +
\ell$ categories. The latter $\ell$ categories are singleton sets; they contain one
leaf each. The former $k$ categories contain the center and the leaves for
which the corresponding bit is one. For any two distinct leaves $i$ and $j$ there is a category $C$ containing $j$ and the center, but not
$i$, because the bitstrings form an anti-chain. We use $C$ to route from $i$ to
the center and the singleton category for $j$ to continue to $j$.   
Observe that the center is in $k$ categories and each leaf is in $\lfloor k/2 \rfloor + 1 \le k$ categories.

The equation for $k$ can be derived using Stirling's formula.
\end{proof}

Similar techniques apply for stars with larger diameters.

\begin{lemma}\label{lem:long_star_upper_bound} For the star with $\ell$ leaves
and diameter $2 d$, there is a system of categories $\calS$ with membership dimension
\[
 \memdim(\calS) \in \Oh(d\log (n/d)).
\]
\end{lemma}

\begin{proof} The Lemma follows from an application of the algorithm from Lemma~\ref{lem:trees}. We also give an alternative proof.

We perform a straightforward extension of the solution for stars of diameter 2. Again every leaf $u$ has a binary number $b(u)$, encoded as categories $O_i$ and $Z_i$, $1\leq i\leq \ceil{\log \ell}$. Create $d$ copies of these numbers using different categories. Category $O_i^{(k)}$, respectively $Z_i^{(k)}$, contains all nodes on the path from leaves $\{u \st b_i(u) = 1\}$, respectively $\{u \st b_i(u) = 0\}$, to the center, as well as all nodes on the remaining paths up to a (graph-)distance $k$ from the center, see Fig.~\ref{fig:long_star}.

Additional $d$ categories are needed for every leaf to route from the center down to the leaves, analogous to the singleton categories for leaves in the small star.\end{proof}

\begin{figure}[t]
	\centering
	\includegraphics[width=0.5\linewidth]{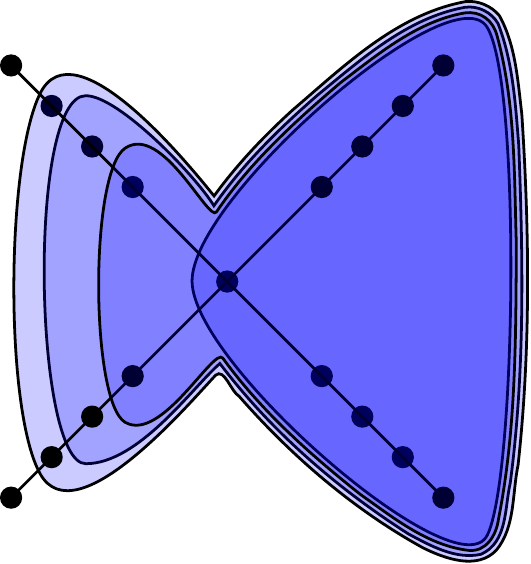}
	\caption{The $d$ copies of a category.}\label{fig:long_star}
\end{figure}

We proceed to show a lower bound. A key ingredient of our proof is the following problem of covering the set
\[
	Q = \{(i,j) \st 1 \le i,j \le \ell,\ i \neq j\}
\]
by $t$ rectangles, i.\,e., 
\[
	Q = \bigcup_{1 \le k \le t} S_k  \times T_k .
\]
Clearly $S_k \cap T_k = \emptyset$ in any covering of $Q$ by
rectangles. We will first prove a lower bound for the rectangle covering
problem and then obtain a lower bound for membership dimension by a reduction
to the rectangle covering problem. The intuition for the reduction is
that every category $C_k$ corresponds to a set $S_k\times T_k$
such that $(i,j)\in S_k\times T_k$ if $C_k$ allows us to route between $i$ and
the center when having $j$ as target. Hence, by bounding the
number of rectangles needed to cover $Q$, we bound the number of categories
that need to contain the center in a star with $\ell$ leaves. 

\begin{lemma} \label{lem:sum_sk} If $4 \ln t \le \ln \ell$, we have
\[
	\sum_k \abs{S_k} \ge \frac{\ell \ln \ell}{32\ln t}
\] 
in any covering of $Q$ by $t$ rectangles. 
\end{lemma}
\begin{proof}The proof uses a double counting technique. For every $i$, $1 \le i \le \ell$, define a binary vector $v_i \in \{0,1\}^t$ that indicates which $S_k$ contain $i$ as
\[
	(v_i)_k = (i \in S_k), \quad 1 \leq k \leq t.
\]
We show that the $v_i$ are pairwise distinct and hence there must be at least $\ell$ such vectors. Consider any $i$ and $j$ with $i \neq j$. Since $(i,j)$ in $Q$, there must be
a $k$ with $i \in S_k$ and $j \in T_k$. Since $S_k \cap T_k = \emptyset$, we
conclude $j \not\in S_k$ and hence $(v_i)_k \neq (v_j)_k$.

Clearly the total number of ones in all $v_i$ equals $\sum_k \abs{S_k}$. This
number is minimized if there is an $h_0$ such that all vectors with less than
$h_0$ ones are used, and the remaining vectors contain exactly $h_0$ ones. Then
$h_0$ must be such that 
\begin{align*}
  \ell &\le \sum_{h=0}^{h_0} {t \choose h} \le \sum_{h=0}^{h_0} t^h \\
	&= \frac{t^{h_0 + 1} - 1 }{t - 1} \\
	&\le t^{h_0 + 1}. 
\end{align*}
Thus as $4\ln t \leq \ln \ell$ by assumption,
\[   
	h_0 \ge \frac{\ln \ell}{\ln t} - 1 \ge \frac{\ln \ell}{2 \ln t}.
\]
We split the sum into the vectors that contain less than $h_0$ ones and the rest. Then the number of ones is at least
\[
	N := \sum_{h=0}^{h_0-1} {t \choose h} h + \left( \ell - \sum_{h=0}^{h_0-1} {t \choose h} \right) h_0.
\]
We now distinguish cases. If 
\[
	\ell - \sum_{h=0}^{h_0-1}\ {t \choose h} \ge \ell/2
\]
then $N \ge (\ell/2) h_0$. Otherwise there are more than $\ell/2$ vectors with less than $h_0$ ones and we lower bound $N$ by the first term. There are two subcases. If
$h_0 \le t/3$ and hence for all $1\leq h \leq h_0$
\[
	{t \choose h}/{t \choose h-1} = \frac{t+1}{h} -1 \ge 2 ,
\] 
we have 
\begin{align*}
	\ell/2 &\le \sum_{h=0}^{h_0-1}{t \choose h} \le  {t \choose h_0 - 1} \sum_{j \ge 0} 2^{-j}\\
		&= 2 {t \choose h_0 - 1}, 
\end{align*}
and hence 
\begin{align*}
	N &\ge \sum_{h=0}^{h_0-1} {t \choose h} h \geq {t \choose h_0-1} (h_0-1)\\
	  &\geq \frac{\ell}{4} (h_0 - 1)\\
	  &\geq \frac{\ell \ln \ell}{16 \ln t}.
\end{align*}
Finally, if $h_0 > t/3$, we bound
\begin{align*}
\sum_{h=0}^{t/4-1}\! { t \choose h}\! &\leq 2{t \choose \frac t4-1} \leq \frac{1}{2} {t \choose \frac t4+1}\\
	&\leq \frac{1}{2} \sum_{h=t/4}^{h_0-1} {t \choose h}
\end{align*}
and hence  
\begin{align*}
	\sum_{h=t/4}^{h_0-1} {t \choose h} &= \sum_{h=0}^{h_0-1} {t \choose h} - \sum_{h=0}^{t/4-1} {t \choose h}\\
		&\geq \frac{1}{2} \sum_{h=0}^{h_0-1} {t \choose h}\\
		&\geq \ell/4.
\end{align*} 
Therefore
\[
	N \ge \frac \ell 4 \cdot \frac t4 \ge \frac{\ell h_0}{16} \geq \frac{\ell \ln \ell}{32 \ln t}. 
\]
\end{proof} 

We now consider a star with $\ell$ leaves and diameter $2 d$. Let
\[
 	\mathcal{P} = \left\{P^{(1)},\ldots,P^{(t)}\right\}
\] 
be the family of categories containing the center. We will lower bound $t$ using Lemma \ref{lem:sum_sk}. 

Every category containing the center corresponds to a partition 
\[
	P = (P_0,P_1,\ldots, P_d)
\]
of $\{1,\ldots,\ell\}$, where the set $P_k$ contains
the leaves for which a path of length $k$ from the center to the leaf is
contained in the category. 

\newcommand{\calP}{\mathcal{P}}

\begin{lemma}\label{lem:partition_distinct} For every $i$ and $j$ with $1 \le i , j \le \ell$, $i \ne j$, and every $k$, $0 \le k < d$, there must be a $P \in \calP$ such that 
\[    
	(i \in P_k) \neq (j \in P_k).
\]
\end{lemma}
\begin{proof} There must be a $P$ such that $i \in P_k$ and $j \in P_d$ as
otherwise we cannot route from leaf $i$ to leaf $j$. Since $P$ is a partition, $j \not\in
P_k$. \end{proof}

\begin{lemma}\label{lower bound lemma} Let $t$ be the number of categories containing the center in a star with $\ell$ leaves and diameter $2 d$. Then
\[
	t \ln t \ge \frac{d \ln \ell}{32}.
\]\end{lemma}
\begin{proof}
This proof uses a double counting technique. For each $1\leq i \leq \ell$ and $k$, define the vector $v_{i,j}$ by
\[ 
	(v_{i,k})_j = (i \in (P^{(j)})_k). 
\]
Then by Lemma~\ref{lem:partition_distinct} for every $k$ the vectors in $\{v_{i,k} \st 1 \le i \le \ell\}$ are
pairwise distinct, and, hence, by Lemma~\ref{lem:sum_sk} for every $k$, $0 \le k < d$, 
\[
	\sum_{1 \le j \le t} \abs{(P^{(j)})_k} \ge \frac{\ell \ln \ell}{32 \ln t}.
\]
Summation over $k$ yields
\begin{align*}   
	d \frac{\ell \ln \ell}{32 \ln t} &\le \sum_{k = 0}^{d-1}  \sum_{j=1}^{t} \abs{(P^{(j)})_k}\\
							   &= \sum_{j=1}^{t} \sum_{k=0}^{d-1}\abs{(P^{(j)})_k}\\
							   &\le t \ell. && \text{\qed}
\end{align*}
\end{proof}

\begin{theorem} Consider a star with $\ell$ leaves and diameter $2 d$, and let $t$ be the number of categories that contain the center. If $\ell \ge 3$, 
\begin{align*}
	t &\ge \frac{d \ln \ell}{32 (\ln d + \lnln \ell)}  \\
	& = \Omega\left( \frac{d \log (n/d)}{\log( d \log(n/d) )} \right).
\end{align*}
\end{theorem}
\begin{proof} Assume otherwise, and let $X = d \ln \ell$. Since $t \ln t$ is an
increasing function in $t$, Lemma~\ref{lower bound lemma} implies
\[   
	\frac{X}{32 \ln X} \ln\left(\frac{X}{32 \ln X}\right) \ge \frac{X}{32}
\]
and hence 
\[   
	\ln\left(\frac{X}{32 \ln X}\right) \ge  \ln X, 
\]
a contradiction. \end{proof}

\section{A Universal Lower Bound}\label{sec:general_lower_bound}

The lower bounds of the preceding section are existential. We showed the
existence of graphs for which every good system of categories has a certain
membership dimension. The lower bound almost matches the universal upper bound
of Sect.~\ref{sec:improved_upper_bound}. In this section, we show a universal
lower bound: every good system of categories for a graph $G$ has membership dimension
\[   \Omega\left(\diam(G) + \frac{\log \abs{V}}{\log \delta}\right),\]
where $\delta$ is the average degree. We also show that this bound is best
possible. 

\begin{theorem}\label{lem:general_case} In a graph $G=(V,E)$ with $n$ nodes and average degree
$\delta/2$, there is a node that is contained in 
\[
\Omega(\diam(G) + \log n/\log \delta)
\] 
categories.
\end{theorem}
\begin{proof} The lower bound of $\diam(G)$ was already established by Eppstein
et al., see Lemma \ref{lem:path_lower_bound}. We turn to the second bound. Since $G$ has average degree $\delta/2$, there are at least $n/2$ nodes that have degree less than $\delta$. Among these we can greedily find an independent set $I$ of size at least $\Omega(n/\delta)$. Consider the graph 
\[
	G'=\left(V(G), E(G) \cup {V(G)\setminus I \choose 2}\right),\]
i.\,e., $G$ where the subgraph outside of $I$ is augmented to a clique. Routing in $G'$ can only be easier than in $G$.
We show a lower bound for $G'$.

We want to show that either a node in $I$ or a node in its neighborhood 
\[
	N(I) = \{u \st v \in I \wedge \{u,v\}\in E\}
\] is in $\Omega(\log n/\log \delta)$ categories. Fix a system of categories
$(C_i)_{1\leq i \leq m}$ that allows greedy routing. For every node $v$ in $I
\cup N(I)$ define a pattern $p^{(v)}$ as 
\[
 \left(p^{(v)}\right)_i\! =\! \begin{cases}
            1 &\! v\in C_i\\
			0 &\! v\not \in C_i \wedge N(v) \cap C_i \neq \emptyset\\
			* &\! v\not \in C_i \wedge N(v) \cap C_i = \emptyset.
           \end{cases}
\]
We say two patterns $p^{(v)}$, $p^{(u)}$ \emph{match} if they agree on all
positions where both have a $1$ or $0$ (i.\,e., $*$ matches to
anything). Let $u$ and $v$ be distinct vertices. Greedy routing from $u$ to $v$ requires the existence of
a neighbor $z$ of $u$ and a category $C_i$ with $C_i \in (\cat(v) \cap \cat(z))
\setminus \cat(u)$. Then $p^{(v)}_i = 1$ and $p^{(u)}_i = 0$. Hence routing is only possible if no two patterns match. 

To each pattern we assign a region of points in the hypercube $\{0,1\}^m$,
namely the set of all matching bitstrings. Observe that two
patterns match iff their regions in the hypercube overlap. Intuitively, this
means that for allowing greedy routing we need to set many values in the
$p^{(v)}$ to 0 or 1 to make these regions small enough to accommodate all
without overlap. A high number of 1's forces nodes from $I$ into many
categories, a high number of 0's forces nodes from $N(I)$ into many categories.

To make this into a formal argument, consider the cost of a vertex $v \in I$ defined as
\[
 c(v) := t_1^{(v)} + \frac{1}{\delta} t_0^{(v)},
\]
where $t_k^{(v)}$ is the number of positions in the vector $p^{(v)}$ equal to
$k$. 
Note that $c(v)/2$ is a lower bound for 
\[
	\max_{u \in \{v\} \cup N(v)} |\cat(u)|,
\]
the maximal number of categories $u$ or one of its neighbors is in, since each
1 means an additional category for $v$ and each 0 an additional category for
one of its at most $\delta$ neighbors. 
In the remainder of the proof we show that there is a vertex $v \in I$ with
$c(v) = \Omega(\log n / \log \delta)$.  

For $\delta = 1$, an easy argument shows an $\Omega(\log n)$ lower bound. 
Define a measure on the hypercube $\{0,1\}^m$ as 
\[
 \mu(x) = 2^{-m},
\]
for $x \in \{0,1\}^m$, and by 
\[
	\mu(X) = \sum_{x \in X} \mu(x)
\]
for $X \subset \{0,1\}^m$. Then the total measure of the hypercube is 1. As the
regions defined by the patterns $p^{(v)}$ with $v \in I$ must be disjoint,
there must be a pattern that has measure $\mu(p^{(v)})$ no more than $1/n$. As
$\mu(p^{(v)}) = 2^{-c(v)}$ we get $c(v) \ge \log n$.  

The same argument can be applied if $\delta$ is bounded by some constant, as then
replacing $\delta$ by 1 in the definition of $c(v)$ does not change it
asymptotically. 

For $\delta$ greater than some large enough constant we change the measure on the hypercube to 
\[
 \hat \mu(x) = \alpha^{\sum x} \cdot (1-\alpha)^{m-\sum x},
\]
for an $0 < \alpha < 1$ that is to be determined, where $\sum x$ denotes the number of 1's in the bitstring $x \in \{0,1\}^m$ (and, thus, $m-\sum x$ the number of 0's). Again we sum up for subsets of the hypercube and the whole hypercube has measure 1. The parameter $\alpha$ will incorporate the reduced weight of 0's in the patterns. We want to choose $\alpha$ such that for some $C > 0$
\begin{align*}
   \alpha   &= 2^{-C}\\
(1-\alpha)  &= 2^{-C/\delta},
\end{align*}
or, equivalently, we want to find a $C$ such that
\begin{align} \label{eq:defc}
 1 = 2^{-C} + 2^{-C/\delta},
\end{align}
as then again $\mu(p^{(v)}) = 2^{-C \cdot c(v)}$. Now, we argue as in the $\delta = 1$ case. Since the patterns may not overlap, there has to be a pattern with measure $\mu(p^{(v)})$ at most $1/n$. As $\mu(p^{(v)}) = 2^{-C \cdot c(v)}$ we get $c(v) \ge \log n/C$, which shows the claim assuming $C = O(\log \delta)$. It remains to show the latter.

Unfortunately, we cannot solve \eqref{eq:defc} exactly for $C$. However,  
for $C = O(\delta)$ we have $2^{-C/\delta} = 1 - \Theta(C/\delta)$,
and, thus, for $C = \log \delta$
\[
  2^{-C} + 2^ {-C/\delta} = 1 + \frac{1 - \Theta(\log \delta)}{\delta} < 1,
\]
for $\delta$ large enough. On the other hand, for $C = \log \delta - 2 \log \log \delta$ we have
\[
  2^{-C}\! + 2^ {-C/\delta} = 1 + \frac{\log^{2} \delta - \Theta(\log \delta)}{\delta} >\! 1,
\]
for $\delta$ large enough. Hence, there is a root $C = \log \delta - \Theta(\log \log \delta)$.
\end{proof}

\begin{figure}[t]
 \centering
 \includegraphics[width=0.5\linewidth]{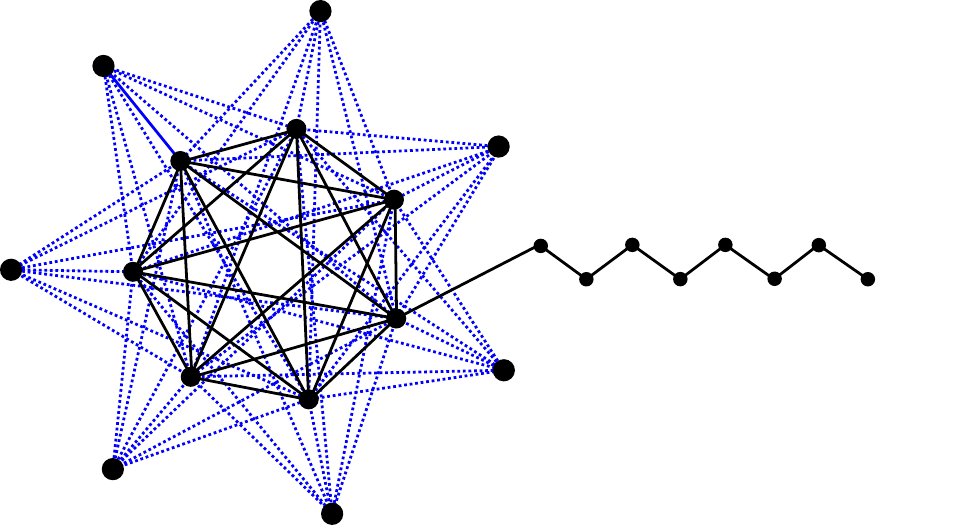}
 \caption{Tight example for Theorem~\ref{lem:general_case}}
 \label{fig:zauberstab}
\end{figure}

This bound is tight. For any triple $(n, \delta, \diam)$ of positive reals with
$1 \le \delta, \diam \le n$, we can construct a
tight example with $\Theta(n)$ vertices, average degree $\Theta(\delta)$, and
diameter $\Theta(\diam)$. For this construction take a $\delta$-Clique
$K_\delta$ and connect a set $O$ of $n$ additional nodes to all nodes in
$K_\delta$. Add a path $P$ of length $\diam$. For an example see
Fig.~\ref{fig:zauberstab}. The graph $G$ thus constructed has the required
parameters. 

We construct a good system of categories for $G$ with a membership dimension matching the lower bound
of $\Theta(\diam(G)+\log(n)/\log(\delta))$ from
Theorem~\ref{lem:general_case}. Routing from the nodes on $P$ to all other
nodes is possible with $\diam(G)$ membership dimension as in the solution
for simple paths (see Fig.~\ref{fig:path_categories}). In the construction we
treat the nodes of $K_\delta$ as one, except for the node that is directly
connected to $P$, and similarly the nodes in $O$.  

To enable routing between the $n$ nodes in $O$, we generalize the construction from
Lemma~\ref{lem:simple_star_upper_bound}. 
We number the vertices in $O$ with base $\delta$ (more precisely,
$\max\{2,\lceil\delta\rceil\}$). Then any vertex $v \in O$ corresponds to a unique string $(b^{(v)}_0,\ldots,b^{(v)}_{k-1})$ of length $k = \log n / \log \delta$
with each $b_i \in \{0, \ldots, \delta-1\}$.
For each $0 \le i < k$ and $0 \le j < \delta$ we create a category
\[
    C_{i,j} := \{ v \in O \mid b^{(v)}_i = j \} \cup \{ u_j \},
\]
where $u_1,\ldots,u_\delta$ are the vertices of the clique $K_\delta$. 
Additionally we add a singleton category for each vertex. 

Now, every vertex in the clique and $O$ is in at most $k+1$ such categories.
Moreover, we can route between any pair $u,v \in O$, as there is an $i$ with 
$b^{(u)}_i \ne b^{(v)}_i$, so category $C_{i,b^{(v)}_i}$ allows to route from $u$
to the clique when having $v$ as target.


\section{Conclusion and Open Problems}

In this paper we presented an improved construction of systems of categories
$\calS$ that support greedy routing in general graphs $G$. The previous best result
uses a membership dimension of $\Oh((\diam(G)+\log \abs{V})^2)$, whereas our methods
show that  
\[
	\Oh(\diam(G)\log (2\abs{V}/\diam(G))
\]
is sufficient. Besides improved upper bounds we also show stronger lower bounds than previously known. Our results improve the lower bound from $\diam(G)$ to 
\[
 \Omega(\diam(G) + \log \abs{V}/ \log \delta)\]
for graphs of average degree $\delta$. This lower bound is tight for certain
graphs. For the restricted class of stars of diameter $2 d$ and $\ell$ leaves, having $n$ nodes, 
we showed a lower bound of 
\[\Omega\left(\frac{d\log (n/d)}{\log (d \log (n/d))}\right).\]
For the case of diameter 2 with $\ell$ leaves the stronger bound 
\[\log \ell + \left(\frac 12 + o(1)\right) \loglog \ell\]
holds.

These bounds are small enough to provide additional evidence that the category-based theory of message passing in social networks is correct.

Many open problems remain. Foremost, we conjecture that $\Oh(d\log(2n/d))$ is tight for stars of diameter $d$, but our lower bound is weaker. Moreover, the best upper bounds are achieved by taking a spanning tree and constructing categories for it. Is it possible to exploit the properties of the graph better than this?

Empirical studies show that social networks are graphs with a power-law degree distribution and a large clustering coefficient (see e.\,g.~\cite{Newman03thestructure} for an overview). The example graphs considered in this paper do not have these properties. It would therefore be interesting to study their consequences in the minimal membership dimension required for routing. Intuitively the membership dimension should also depend on the expansion properties of the graph.

From a sociological point of view it is also interesting to see how natural relaxations of the greedy rule, e.\,g., allowing nodes to choose a neighbor at random in case of distance ties, influence the bounds. 

\bibliographystyle{elsarticle-num}
\bibliography{bibliography}
\end{document}